\documentclass[a4paper,10pt]{elsarticle}
\usepackage[english]{babel}
\usepackage[T1]{fontenc}
\usepackage{amsmath, amsthm}
\usepackage{ae,aecompl}
\usepackage{amssymb}
\usepackage[section] {placeins}
\usepackage[utf8]{inputenc}

\usepackage[ruled, vlined, linesnumbered]{algorithm2e}
\newtheorem {Theorem}                 {Theorem}         [section]
\newtheorem {theorem}      [Theorem]  {Theorem}
\newtheorem {myalgorithm}    [Theorem]  {Algorithm}

\newtheorem {lemma}        [Theorem]  {Lemma}

\input amssym.def
\input amssym

\usepackage{pgfplots}
\usepackage{verbatim}
\usepackage{subfigure}
\usepackage{tikz}
\usetikzlibrary{shapes,arrows,backgrounds,mindmap}
\usepackage{titlesec}
\journal{arXiv}

\begin{document}
	\begin{frontmatter}
		\title{ The problem of computing a $2$-T-connected spanning subgraph with minimum number of edges in directed graphs}
		\author{Raed Jaberi}
		\ead{Jaberi.Raed@gmail.com}
	    \author{Reham Mansour}
	    \ead{rehammansor1698@gmail.com‏}
	   
		\begin{abstract}
		 
		 Let $G=(V,E)$ be a strongly connected graph with $|V|\geq 3$. For $T\subseteq V$, the strongly connected graph $G$ is $2$-T-connected if $G$ is $2$-edge-connected and for each vertex $w$ in $T$, $w$ is not a strong articulation point.
This concept generalizes the concept of $2$-vertex connectivity when $T$ contains all the vertices in $G$. This concept also generalizes the concept of $2$-edge connectivity when $|T|=0$. The concept of $2$-T-connectivity was introduced by Durand de Gevigney and Szigeti in $2018$. In this paper, we prove that there is a polynomial-time 4-approximation algorithm for the following problem: given a $2$-T-connected graph $G=(V,E)$, identify a subset $E^ {2T} \subseteq E$ of minimum cardinality such that $(V,E^{2T})$ is $2$-T-connected. 
	\end{abstract} 
	\end{frontmatter}
	\section{Introduction}
Let $G=(V,E)$ be a strongly connected graph, where $n=|V|$ and $m=|E|$. A vertex $w \in V$ is a strong articulation point if the subgraph $G \setminus \lbrace w\rbrace$ is no longer strongly connected. An edge $(u,w) \in E$ is a strong bridge if the subgraph $(V,E\setminus \lbrace (u,w)\rbrace)$ is no longer strongly connected. The strongly connected graph $G$ is $2$-vertex connected if $|V|\geq 3$ and for each vertex $w \in V$, $w$ is not a strong articulation point. We say that $G$ is $2$-edge connected if for every edge $e \in E$, $e$ is not a strong bridge. For $T\subseteq V$, the strongly connected graph $G$ is called $2$-T-connected if $G$ is $2$-edge connected and no vertex in $ T$ is a strong articulation point \cite{GevigneySzigeti2018}.  This concept generalizes the concept of $2$-vertex connectivity when $T$ contains all the vertices in $G$. This concept also generalizes the concept of $2$-edge connectivity when $|T|=0$ \cite{GevigneySzigeti2018}.

In this paper we consider the following problem, denoted by M2TC: given a $2$-T-connected graph $G=(V,E)$, identify a subset $E^ {2T} \subseteq E$ of minimum cardinality such that $(V,E^{2T})$ is $2$-T-connected. 

Let $G=(V,E)$ be a strongly connected graph and let $u \in V$. Obviously, $G(u)=(V,E,u)$ is a flowgraph with start vertex $u$ since for all vertices $w\in V$, there exists a path from the start vertex $u$ to $w$. A vertex $x\in V$ is a dominator of vertex $y$ in $G(u)$ if $x$ lies on all paths from the start vertex $u$ to $y$ in $G(u)$. A vertex $x \in V\setminus \lbrace u\rbrace$ is a non-trivial dominator in $G(u)$ if $x$ is a dominator of some vertex $y \in V\setminus \lbrace u,x\rbrace$ \cite{ItalianoLauraSantaroni2012}. Two spanning trees $T_{u}^{1},T_{u}^{2}$ of $G(u)$ rooted at the start vertex $u$ in $G(u)$ are independent if for any vertex $x\in V\setminus \lbrace u\rbrace$, the paths from $u$ to $x$ in $T_{u}^{1}$ and $T_{u}^{2}$ have only the dominators of $x$ in common \cite{GeorgiadisTarjan2005,Georgiadis12011}.

In $2010$, Georgiadis \cite{Georgiadis2010} showed how to check whether a given strongly connected graph is $2$-vertex connected in $O(n+m)$ time. Italiano et al. \cite{ItalianoLauraSantaroni2012} gave a linear time algorithm for computing all the strong articulation point of a directed graph. The idea behind this algorithm \cite{ItalianoLauraSantaroni2012} comes from the following important property: Let $G=(V,E)$ be a strongly connected graph and let $u \in V$. A vertex $w\in V\setminus \lbrace u\rbrace$ is a strong articulation point if and only if $w$ is a non-trivial dominator in $G(u)$ or in $G^{R}(u)=(V,E^{R},u)$, where $E^{R}$ denotes the set $\lbrace (a,b) \mid (b,a) \in E\rbrace$. Moreover, Italiano et al. \cite{ItalianoLauraSantaroni2012} gave linear time algorithms for finding strong bridges. 
In \cite{FirmaniGeorgiadisItalianoLauraSantaroni2016}, Firmani et al. designed and engineered all algorithms presented in \cite{ItalianoLauraSantaroni2012}. In $2018$, Durand de Gevigney and Szigeti \cite{GevigneySzigeti2018} introduced the concept of $2$-T-connectivity and proved that every minimal $2$-T-connected graph contains a vertex with indegree and outdgree $2$.
The problem of calculating a $k$-vertex-connected spanning subgraph with minimum number of edges in directed graphs is NP-hard \cite{GareyJohnson1979}. Since this problem is a special case of M2TC when $T=V$ and $k=2$, then M2TC is also NP-hard. By results from \cite{Edmonds1972,Mader1985}, the cardinality of the edge set of each minimal $k$-vertex-connected directed graph is at most $2kn$ edges \cite{CheriyanTThurimella2000}. 
Cheriyan et al. \cite{CheriyanTThurimella2000} provided a polynomial time algorithm that achieves factor of $(1+1/k)$ for the minimum cardinality $k$-vertex-connected spanning subgraph problem. Georgiadis et al. \cite{GeorgiadisItalianoKaranasiou2020,GeorgiadisItalianoKaranasiou2017,Georgiadis12011} gave efficient approximation algorithms for the minimum cardinality $2$-vertex connected spanning subgraph problem in directed graphs and improved the algorithm of Cheriyan et al. \cite{CheriyanTThurimella2000}. 
Jaberi \cite{Jaberi2016} studied three problems which are generalization of the minimum cardinality $2$-vertex connected spanning subgrpah problem. Furthermore, Jaberi\cite{Jaberi2015} presented algorithms for the problem of computing $2$-blocks in directed graphs and approximation algorithms for the problem of computing a minimum size strongly connected spanning subgraph with the same $2$-blocks. Georgiadis et al. \cite{GeorgiadisGeorgiadisLauraParotsidis2015,GeorgiadisItalianoLauraParotsidis2016,
GeorgiadisItalianoLauraParotsidis15,GeorgiadisItalianoLauraParotsidis2018,
GeorgiadisItalianoLauraParotsidisKaranasiouPaudel2018,
GeorgiadisItalianoPapadopoulosPapadopoulosParotsidis2015} provided efficient algorithms for solving these problems. Jaberi \cite{Jaberi2021,Jaberi2edgetwinlessblocks,JaberiComputing2twinlessblocks} presented approximation algorithms for the minimum $2$-vertex strongly biconnected spanning subgraph problem and algorithms for computing $2$-twinless blocks and $2$-edge-twinless blocks.
In this paper, we prove that there is a $4$-approximation algorithm for M2TC

\section{$4$-Approximation algorithm for M2TC}
In this section we describe an approximation algorithm for M2TC which is based on strong articulation points \cite{ItalianoLauraSantaroni2012} ,independent spanning trees \cite{GeorgiadisTarjan2005}, and the results of Georgiadis et al. \cite{GeorgiadisItalianoKaranasiou2020,GeorgiadisItalianoKaranasiou2017,Georgiadis12011}. 

Note that each $2$-T-connected graph is $2$-edge-connected, but the converse is not necessarily true. Therefore, optimal solutions for minimum $2$-edge-connected spanning subgraph problem are not necessarily $2$-T-connected  as illustrated in Figure \ref{figure:exampleM2TCandItsOptimalSolution}
\begin{figure}[htbp]
	\centering

\subfigure[]{	
\scalebox{0.81}{

		\begin{tikzpicture}[xscale=1.9]
		\tikzstyle{every node}=[color=black,draw,circle,minimum size=0.9cm]
		\node (v1) at (-1.9,3.1) {$1$};
		\node (v10) at (1.9,-1.5) {$10$};
		\node (v2) at (-2.4,0) {$2$};
		\node (v3) at (0, -2.5) {$3$};
		\node (v4) at (3.4,-2.1) {$4$};
		\node (v5) at (1.4,0.1) {$5$};
		\node (v6) at (3.7,0.7) {$6$};
		\node (v7) at (3.7,3.4) {$7$};
		\node (v8) at (0.1,3.5) {$8$};
		\node (v9) at (-0.7,0) {$9$};
		\node (v13) at (-2,-2.5) {$13$};
		\node (v11) at (1.9,2.5){$11$};
		\node (v12) at (-0.7,-2.5) {$12$};
	     
		\begin{scope}   
		\tikzstyle{every node}=[auto=right] 
		\draw [-triangle 45] (v6) to (v10);  
		\draw [-triangle 45] (v2) to [bend left](v6);
		\draw [-triangle 45] (v4) to (v10);
		\draw [-triangle 45] (v10) to [bend right](v13);
	     \draw [-triangle 45] (v10) to[bend right]  (v4);
	  
		\draw [-triangle 45] (v3) to (v5);
	     \draw [-triangle 45] (v5) to[bend right]  (v3);
	     \draw [-triangle 45] (v10) to (v3);
	     \draw [-triangle 45] (v3) to[bend right]  (v10);
	     
\draw [-triangle 45] (v8) to[bend right] (v11);
	     \draw [-triangle 45] (v11) to  (v8);
	     \draw [-triangle 45] (v11) to (v7);
	     \draw [-triangle 45] (v7) to[bend right]  (v11);

		\draw [-triangle 45] (v2) to (v1);
	     \draw [-triangle 45] (v1) to[bend right]  (v2);
	     \draw [-triangle 45] (v12) to (v13);
	     \draw [-triangle 45] (v13) to[bend right]  (v12);	
	     
	     \draw [-triangle 45] (v8) to[bend right]  (v1);
	     \draw [-triangle 45] (v1) to (v8);
	     \draw [-triangle 45] (v8) to [bend right] (v9);
	     \draw [-triangle 45] (v9) to (v8);	
	     
	         \draw [-triangle 45] (v8) to[bend right]  (v5);
	     \draw [-triangle 45] (v5) to (v8);
	     \draw [-triangle 45] (v12) to [bend right] (v9);
	     \draw [-triangle 45] (v9) to (v12);	
	         \draw [-triangle 45] (v7) to[bend right]  (v6);
	     \draw [-triangle 45] (v6) to (v7);
	     \draw [-triangle 45] (v6) to  (v4);
	     \draw [-triangle 45] (v4) to [bend right] (v6);	
	     
	      \draw [-triangle 45] (v2) to[bend right] (v13);
	     \draw [-triangle 45] (v13) to  (v2);

		\end{scope}
		\end{tikzpicture}
	}
	}	
	\subfigure[]{
\scalebox{0.81}{

		\begin{tikzpicture}[xscale=1.9]
		\tikzstyle{every node}=[color=black,draw,circle,minimum size=0.9cm]
		\node (v1) at (-1.9,3.1) {$1$};
		\node (v2) at (-2.4,0) {$2$};
		\node (v3) at (0, -2.5) {$3$};
		\node (v4) at (3.4,-2.1) {$4$};
		\node (v5) at (1.4,0.1) {$5$};
		\node (v6) at (3.7,0.7) {$6$};
		\node (v7) at (3.7,3.4) {$7$};
		\node (v8) at (0.1,3.5) {$8$};
		\node (v9) at (-0.7,0) {$9$};
		\node (v10) at (1.9,-1.5) {$10$};
		\node (v11) at (1.9,2.5){$11$};
		\node (v12) at (-0.7,-2.5) {$12$};
	     \node (v13) at (-2,-2.5) {$13$};
		\begin{scope}   
		\tikzstyle{every node}=[auto=right] 
	
		\draw [-triangle 45] (v4) to (v10);
	
	     \draw [-triangle 45] (v10) to[bend right]  (v4);
	  
		\draw [-triangle 45] (v3) to (v5);
	     \draw [-triangle 45] (v5) to[bend right]  (v3);
	     \draw [-triangle 45] (v10) to (v3);
	     \draw [-triangle 45] (v3) to[bend right]  (v10);
	     
\draw [-triangle 45] (v8) to[bend right] (v11);
	     \draw [-triangle 45] (v11) to  (v8);
	     \draw [-triangle 45] (v11) to (v7);
	     \draw [-triangle 45] (v7) to[bend right]  (v11);

		\draw [-triangle 45] (v2) to (v1);
	     \draw [-triangle 45] (v1) to[bend right]  (v2);
	     \draw [-triangle 45] (v12) to (v13);
	     \draw [-triangle 45] (v13) to[bend right]  (v12);	
	     
	     \draw [-triangle 45] (v8) to[bend right]  (v1);
	     \draw [-triangle 45] (v1) to (v8);
	     \draw [-triangle 45] (v8) to [bend right] (v9);
	     \draw [-triangle 45] (v9) to (v8);	
	     
	         \draw [-triangle 45] (v8) to[bend right]  (v5);
	     \draw [-triangle 45] (v5) to (v8);
	     \draw [-triangle 45] (v12) to [bend right] (v9);
	     \draw [-triangle 45] (v9) to (v12);	
	         \draw [-triangle 45] (v7) to[bend right]  (v6);
	     \draw [-triangle 45] (v6) to (v7);
	     \draw [-triangle 45] (v6) to  (v4);
	     \draw [-triangle 45] (v4) to [bend right] (v6);	
	     
	      \draw [-triangle 45] (v2) to[bend right] (v13);
	     \draw [-triangle 45] (v13) to  (v2);

		\end{scope}
		\end{tikzpicture}
	}
	}

\subfigure[]{
	
\scalebox{0.81}{

		\begin{tikzpicture}[xscale=1.9]
		\tikzstyle{every node}=[color=black,draw,circle,minimum size=0.9cm]
		\node (v1) at (-1.9,3.1) {$1$};
		\node (v2) at (-2.4,0) {$2$};
		\node (v3) at (0, -2.5) {$3$};
		\node (v4) at (3.4,-2.1) {$4$};
		\node (v5) at (1.4,0.1) {$5$};
		\node (v6) at (3.7,0.7) {$6$};
		\node (v7) at (3.7,3.4) {$7$};
		\node (v8) at (0.1,3.5) {$8$};
		\node (v9) at (-0.7,0) {$9$};
		\node (v10) at (1.9,-1.5) {$10$};
		\node (v11) at (1.9,2.5){$11$};
		\node (v12) at (-0.7,-2.5) {$12$};
	     \node (v13) at (-2,-2.5) {$13$};
		\begin{scope}   
		\tikzstyle{every node}=[auto=right] 
		  
		\draw [-triangle 45] (v2) to [bend left](v6);
		\draw [-triangle 45] (v4) to (v10);
		\draw [-triangle 45] (v10) to [bend right](v13);
	     \draw [-triangle 45] (v10) to[bend right]  (v4);
	  
		\draw [-triangle 45] (v3) to (v5);
	     \draw [-triangle 45] (v5) to[bend right]  (v3);
	     \draw [-triangle 45] (v10) to (v3);
	     \draw [-triangle 45] (v3) to[bend right]  (v10);
	     
\draw [-triangle 45] (v8) to[bend right] (v11);
	     \draw [-triangle 45] (v11) to  (v8);
	     \draw [-triangle 45] (v11) to (v7);
	     \draw [-triangle 45] (v7) to[bend right]  (v11);

		\draw [-triangle 45] (v2) to (v1);
	     \draw [-triangle 45] (v1) to[bend right]  (v2);
	     \draw [-triangle 45] (v12) to (v13);
	     \draw [-triangle 45] (v13) to[bend right]  (v12);	
	     
	     \draw [-triangle 45] (v8) to[bend right]  (v1);
	     \draw [-triangle 45] (v1) to (v8);
	     \draw [-triangle 45] (v8) to [bend right] (v9);
	     \draw [-triangle 45] (v9) to (v8);	
	     
	         \draw [-triangle 45] (v8) to[bend right]  (v5);
	     \draw [-triangle 45] (v5) to (v8);
	     \draw [-triangle 45] (v12) to [bend right] (v9);
	     \draw [-triangle 45] (v9) to (v12);	
	         \draw [-triangle 45] (v7) to[bend right]  (v6);
	     \draw [-triangle 45] (v6) to (v7);
	     \draw [-triangle 45] (v6) to  (v4);
	     \draw [-triangle 45] (v4) to [bend right] (v6);	
	     
	      \draw [-triangle 45] (v2) to[bend right] (v13);
	     \draw [-triangle 45] (v13) to  (v2);

		\end{scope}
		\end{tikzpicture}
	}
	}
\caption{(a) A $2$-T-connected graph $G$ for T$=\lbrace 8\rbrace$. (b) A $2$-edge-connected subgraph of the graph $G$. Note that this subgraph is not $2$-T-connected for T$=\lbrace 8\rbrace$. (c) An optimal solution for M2TC.}
\label{figure:exampleM2TCandItsOptimalSolution}
\end{figure}
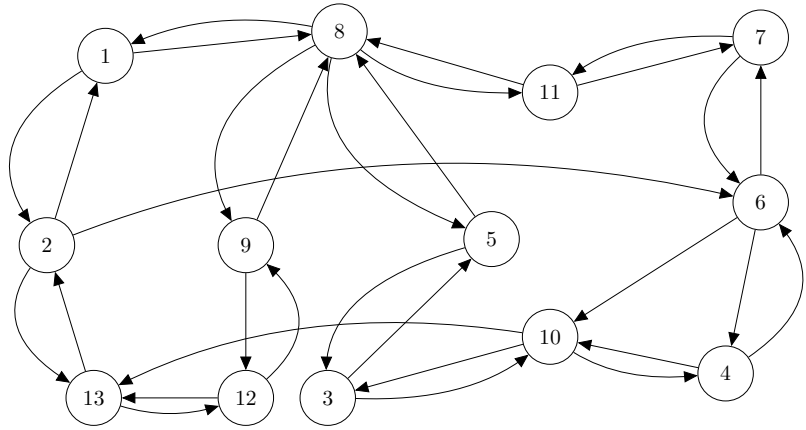
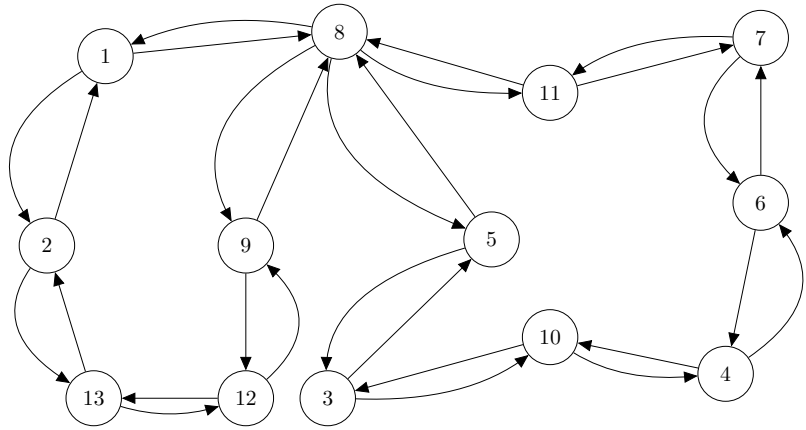
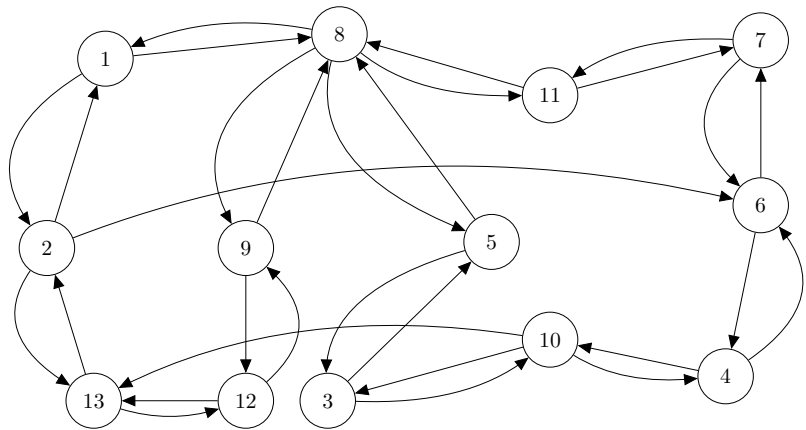

\begin{figure}[ht]
	\begin{myalgorithm}\label{algo:approximationalgorithmfor2Tconnected}\rm\quad\\[-5ex]
		\begin{tabbing}
			\quad\quad\=\quad\=\quad\=\quad\=\quad\=\quad\=\quad\=\quad\=\quad\=\kill
			\textbf{Input:} A $2$-T-connected graph $G=(V,E)$ \\
			
			\textbf{Output:} a $2$-T-connected subgraph $G_{\alpha}=(V,E_{\alpha})$\\
			{\small 1}\> \textbf{if} $T=\emptyset$ \textbf{then}\\
            {\small 2}\>\> find a minimal $2$-edge-connected subgraph $G_{\alpha}=(V,E_{\alpha})$ \\	
             {\small 3}\> \textbf{else} \textbf{if} $T=V$ \textbf{then} \\
             {\small 4}\>\> calculate a minimal $2$-vertex-connected subgraph $G_{\alpha}=(V,E_{\alpha})$ \\
             {\small 5}\> \textbf{else}\\		
             {\small 6}\>\> calculate a minimal $2$-edge-connected subgraph $G_{\alpha}=(V,E_{\alpha})$\\
            {\small 7}\>\> \textbf{if} $|T|>0$ and $|T|\leq 2$ \textbf{then}\\
            {\small 8}\>\>\> \textbf{for} each vertex $y \in T$ \textbf{do}\\
            {\small 9}\>\>\>\> choose a vertex $q \in V\setminus \lbrace y\rbrace$\\
            {\small 10}\>\>\>\>  find a spanning tree $T^{1}$ rooted at $q$ of $G\setminus \lbrace y\rbrace$\\
            {\small 11}\>\>\>\>  \textbf{for} each edge $e\in T^{1}$ \textbf{do} \\         
            {\small 12}\>\>\>\>\>  $E_{\alpha} \rightarrow E_{\alpha}\cup \lbrace e\rbrace$\\
            {\small 13}\>\>\>\>  find a spanning tree $T^{2}$ rooted at $q$ of $G^{R}\setminus \lbrace y\rbrace$, where $G^{R}=(V,E^{R})$\\
            {\small 14}\>\>\>\>  \textbf{for} each edge $(v,w)\in T^{2}$ \textbf{do} \\         
            {\small 15}\>\>\>\>\>  $E_{\alpha} \rightarrow E_{\alpha}\cup \lbrace (w,v)\rbrace$\\
            {\small 16}\>\> \textbf{else}\\
           
			{\small 17}\>\>\> select a vertex $w\in V \setminus T$.\\
			{\small 18}\>\>\> determine two independent trees $T_{1}^{w},T_{2}^{w}$ rooted at $w$ of $G$\\
			 {\small 19}\>\>\>\>  \textbf{for} each edge $e\in T_{1}^{w} \cup T_{2}^{w}$ \textbf{do} \\         
            {\small 20}\>\>\>\>\>  $E_{\alpha} \rightarrow E_{\alpha}\cup \lbrace e\rbrace$\\
			{\small 21}\>\>\>   determine two independent trees $T_{1}^{R},T_{2}^{R}$ rooted at $w$ of  $G^{R}=(V,E^{R})$\\
			{\small 22}\>\>\>  $E_{\alpha} \rightarrow T_{1}^{w} \cup T_{2}^{w}$ \\
			      {\small 23}\>\>\>\>  \textbf{for} each edge $(v,w)\in T_{1}^{R}\cup T_{2}^{R}$ \textbf{do} \\         
            {\small 24}\>\>\>\>\>  $E_{\alpha} \rightarrow E_{\alpha}\cup \lbrace (w,v)\rbrace$\\
		{\small 25}\>output the subgraph $G_{\alpha}=(V,E_{\alpha})$
	
		\end{tabbing}
	\end{myalgorithm}
\end{figure}

The following lemma shows that the subgraph returned by Algorithm \ref{algo:approximationalgorithmfor2Tconnected} is $2$-T-connected.

\begin{lemma} \label{def:outputis2Tconnected}
The output of Algorithm \ref{algo:approximationalgorithmfor2Tconnected} is a feasible solution for M2TC.
\end{lemma}
\begin{proof}
Let $G_{\alpha}=(V,E_{\alpha})$ be the output of Algorithm \ref{algo:approximationalgorithmfor2Tconnected}. We consider three cases: 
\end{proof} 
\begin{itemize} 
\item Case $1$. $T=\emptyset$. In this case, the output $G_{\alpha}$ is $2$-edge connected since it contains a minimal $2$-edge connected subgraph. Thus, by definition, $G_{\alpha}$ is also $2$-T-connected.
\item Case $2$.  $T=V$ . In this case, the output $G_{\alpha}$ is $2$-edge connected since $G_{\alpha}$  contains a minimal $2$-vertex connected subgraph of $G$. Furthermore, for each vertex $x\in T$, the subgraph $G_{\alpha} \setminus \lbrace x\rbrace$ is strongly connected because $G_{\alpha}$ is $2$-vertex-connected.
\item Case $3$. $T\neq V$ and $T\neq \emptyset$. We consider two subcases. In the first subcase we have $|T|>0$ and $T\leq 2$. Note that Algorithm \ref{algo:approximationalgorithmfor2Tconnected} finds in lines $8$--$15$ a strongly connected spanning subgraph of $G\setminus \lbrace y\rbrace$ for each vertex $y \in T$ by computing a spanning tree $T^{1}$ rooted at $q$ of $G\setminus \lbrace y\rbrace$ and a spanning tree $T^{2}$ rooted at $q$ of $G^{R}\setminus \lbrace y\rbrace$ for some vertex $q\in V\setminus\lbrace y\rbrace$. It is easy to see why  the subgrpah obtained by executing lines $9$--$15$ is strongly connected \cite{SandersSandersDietzfelbingerDementiev2019}. In the second subcase we have $|T|>2$ and $|T|<n$ (see lines $17$--$24$). By \cite[Theorem $5.2$]{ItalianoLauraSantaroni2012} A vertex $x\in V\setminus \lbrace w\rbrace$ is a strong articulation point if and only if $x$ is a non-trivial dominator in $G(w)$ or in $G^{R}(w)=(V,E^{R},w)$. The flowgraph $G(w)$ and the flowgraph $(V,T_{1}^{w}\cup T_{2}^{w},w)$  have the same non-trivial dominators \cite{Georgiadis12011,GeorgiadisTarjan2005}. Moreover, the flowgraph $G^{R}(w)$ the flowgraph  $(V,T_{1}^{R} \cup T_{2}^{R},w)$ have also the same non-trivial dominators \cite{Georgiadis12011,GeorgiadisTarjan2005}. Therefore, none of the vertices in T is a strong articulation point in the output $G_{\alpha}=(V,E_{\alpha})$.
\end{itemize}
Theorem \ref{def:approximationfactor4} shows that there exists a $4$-approximation algorithm for M2TC. 
\begin{theorem} \label{def:approximationfactor4}
The subgraph returned by Algorithm \ref{algo:approximationalgorithmfor2Tconnected} has at most $8n$ edges,
\end{theorem}
\begin{proof}
Let $G_{\alpha}=(V,E_{\alpha})$ be the output of Algorithm \ref{algo:approximationalgorithmfor2Tconnected}. We distinguish three cases:
\begin{itemize} 
\item Case $1$. $T=\emptyset$. In this case $G_{\alpha}=(V,E_{\alpha})$ is a minimal $2$-edge connected subgraph of $G$. Results of Edmonds \cite{Edmonds1972} and Mader \cite{Mader1985} imply that  $|E_{\alpha}|$ is not greater than $4n$ \cite{CheriyanTThurimella2000,Georgiadis12011}.
\item Case $2$. $T=V$. In this case $G_{\alpha}=(V,E_{\alpha})$ is a minimal $2$-vertex connected subgraph of $G$. Results of Edmonds \cite{Edmonds1972} and Mader \cite{Mader1985} imply the number of edges in $|E_{\alpha}|$ is at most $4n$. \cite{CheriyanTThurimella2000,Georgiadis12011}.
\item Case $3$. $T\neq V$ and $T\neq \emptyset$. We consider two subcases:
\begin{enumerate}
\item Case $3(a)$. $|T|>0$ and $T\leq 2$. The for loop of lines $8$--$15$ computes a spanning tree $T^{1}$  of $G\setminus \lbrace y\rbrace$ and a spanning tree $T^{2}$ of $G^{R}\setminus \lbrace y\rbrace$ for each vertex $y \in T$. Each spanning tree has only $n-2$ edges. Since $|T|<3$, the number of edges in the edge set computed in lines $8$--$15$ is $2|T|(n-2)\leq 4(n-2)$.
\item Case $3(b)$. $|T|>2$ and $|T|<n$. Lines $17$--$24$ produce two independent trees $T_{1}^{w},T_{2}^{w}$ rooted at $w$ of $G$ and two independent trees $T_{1}^{R},T_{2}^{R}$ rooted at $w$ of  $G^{R}$.  Since each tree of them has only $n-1$ edges, the number of edges in the edge set computed in lines $17$--$24$ is at most $4(n-1)$.
\end{enumerate}
\end{itemize}
Because the edge set of the subgraph computed in line $6$ contains at most $4n$, the number of edges of the output $G_{\alpha}=(V,E_{\alpha})$ is not greater than $8n$.

\end{proof}

\begin{Theorem}
The running time of Algorithm \ref{algo:approximationalgorithmfor2Tconnected} is $O(m^{2})$.
\end{Theorem}
\begin{proof}
	Line $2$ takes $O(m^{2})$ time because testing whether a directed graph is $2$-edge-connected can be done in linear time using the algorithms of Italiano et al. \cite{ItalianoLauraSantaroni2012}. A minimal $2$-vertex-connected subgraph can be identified in time $O(n^2)$ \cite{Georgiadis12011,GeorgiadisItalianoKaranasiou2020}. The for loop of lines $8$--$15$ takes $O(|T|(n+m))=O(n+m)$ since $|T|<3$. Furthermore, lines $17$--$24$ take $O(n+m) $ because two independent spanning trees can be computed in $O(n+m)$ time. \cite{GeorgiadisTarjan2005,Georgiadis12011}.
\end{proof}

\section{Minimal 2-T-connected directed graphs}

A directed graph $G=(V,E)$ is called a minimal $2$-T-connected graph if $G$ is $2$-T-connected but the subgraph $(V,E \setminus \lbrace e\rbrace$ is no longer $2$-T-connected for all $e \in E$ \cite{GevigneySzigeti2018}.

\begin{lemma} \label{def:minimal2tappalgorithm}
If there is some constant $r>1$ such that each minimal $2$-T-connected graph has at most $rn$ edges, then there is a $r/2$-approximation algorithm for M2TC. 
\end{lemma}
\begin{proof}
Let $G=(V,E)$ be a $2$-T-connected graph and let $E_{O}\subseteq E$ be an optimal solution for M2TC. By definition, the subgraph $(V,E_{O})$ is $2$-edge connected because this subgraph is $2$-T-connected.  Since by Menger's Theorem \cite{Menger1927} for edge connectivity, $(V,E_{O})$ contains two edge disjoint paths from $v$ to $w$ for any two vertices $v,w \in V$, each vertex $w\in V$  has indegree at least $2$ in the subgraph  $(V,E_{O})$. Therefore, we have $|E_{O}|\geq 2n$.
\end{proof}

Mader\cite{Mader1978,Mader2002} showed that every minimal $2$-edge connected directed graph and every minimal $2$-vertex connected graph has a vertex with indegree and outdegree $2$. 
Durand de Gevigney and Szigeti \cite{GevigneySzigeti2018} proved that every minimal $2$-T-connected graph has a vertex with indegree and outdegree $2$. By results from \cite{Edmonds1972,Mader1985}, the cardinality of the edge set of each minimal $2$-vertex-connected directed graph is at most $4n$ edges \cite{CheriyanTThurimella2000}. Therefore, there are $2$-approximation algroithms for the minimum $2$-edge connected subgraph problem and the minimum $2$-vertex connected subgraph problem in directed graphs \cite{CheriyanTThurimella2000}. An important question is whether the cardinality of the edge set of each minimal $2$-T-connected graph is at most $4n$.

\section{An improved version of Algorithm \ref{algo:approximationalgorithmfor2Tconnected}}
The subgraph returned by Alogirthm \ref{algo:approximationalgorithmfor2Tconnected} is $2$-T-connected but not necessary a minimal $2$-T-connected graph. Algorithm \ref{algo:improvedapproximationalgorithmfor2Tconnected} is a refined version of Algorithm \ref{algo:approximationalgorithmfor2Tconnected}. Actually, in lines $2$--$5$ of Algorithm \ref{algo:improvedapproximationalgorithmfor2Tconnected} we remove unneeded edges from the output of Alogirthm \ref{algo:approximationalgorithmfor2Tconnected}. Note that the approximation factor of Algorithm \ref{algo:improvedapproximationalgorithmfor2Tconnected} is $4$ because the for loop of lines $2$--$4$ does not add any edge to the output of.Alogirthm \ref{algo:approximationalgorithmfor2Tconnected}.
\begin{figure}[htbp]
	\begin{myalgorithm}\label{algo:improvedapproximationalgorithmfor2Tconnected}\rm\quad\\[-5ex]
		\begin{tabbing}
			\quad\quad\=\quad\=\quad\=\quad\=\quad\=\quad\=\quad\=\quad\=\quad\=\kill
			\textbf{Input:} A $2$-T-connected graph $G=(V,E)$ \\
			\textbf{Output:} a $2$-T-connected subgraph $G_{\alpha}=(V,E_{\alpha})$\\
            {\small 1}\> lines $1$--$24$ of Algorithm \ref{algo:approximationalgorithmfor2Tconnected} \\
            {\small 2}\> \textbf{for} each edge $e \in E_{\alpha}$ \textbf{do} \\
            {\small 3}\>\> \textbf{if} the subgraph $(V,E_{\alpha}\setminus \lbrace e\rbrace)$ is $2$-T-connected \textbf{then} \\
            {\small 4}\>\>\>$E_{\alpha}\rightarrow E_{\alpha}\setminus \lbrace e\rbrace$.\\
		{\small 5}\>output the subgraph $G_{\alpha}=(V,E_{\alpha})$
	
		\end{tabbing}
	\end{myalgorithm}
\end{figure}
\begin{Theorem}
 Algorithm \ref{algo:improvedapproximationalgorithmfor2Tconnected} takes $O(m^{2})$ time.
\end{Theorem}
\begin{proof}
	 We can test whether a directed graph is $2$-T-connected or not in $O(n+m)$ time by computing all the strong bridges and the strong articulation points using the algorithms of Italiano et al. \cite{ItalianoLauraSantaroni2012}. Consequently, lines $2$--$5$ take $O(n(n+m))$ time.
\end{proof}
Note that for each minimal $2$-T-connected graph $G=(V,E)$, Algorithm \ref{algo:improvedapproximationalgorithmfor2Tconnected} returns the same graph. This means each minimal $2$-T-connected graph has at most $8n$ edges.

\end{document}